\def\be{\begin{equation}}
	\def\ee{\end{equation}}
\def\ba{\begin{array}}
	\def\ea{\end{array}}

\def\mathbi#1{\text{\em #1}}

\documentclass[prl,showpacs,twocolumn,amsmath]{revtex4}
\usepackage{amsmath}
\usepackage{amsfonts}
\usepackage{mathrsfs}
\usepackage{amssymb}
\usepackage{pifont}
\usepackage{epsfig,subfigure,dsfont,amsthm,amsbsy,mathrsfs,amscd}
\usepackage{epstopdf}
\usepackage{bbm}
\usepackage{color}
\def\qed{\leavevmode\unskip\penalty9999 \hbox{}\nobreak\hfil
	\quad\hbox{\leavevmode  \hbox to.77778em{%
			\hfil\vrule   \vbox to.675em%
			{\hrule width.6em\vfil\hrule}\vrule\hfil}}
	\par\vskip3pt}
\usepackage{leftidx}
\newtheorem{theorem}{Theorem}

\input amssym.def
\begin{document}
	\title{\large\bf Quantifying quantum-state texture}
	\author{Yiding Wang, Hui Liu, and Tinggui Zhang$^{\dag}$ }
	\affiliation{ School of Mathematics and Statistics, Hainan Normal University, Haikou, 571158, China \\
		 \\ $^{\dag}$ Correspondence to tinggui333@163.com}
	
	\bigskip
	\bigskip
	
	\begin{abstract}
	Quantum-state texture is a newly recognized quantum resource that has garnered attention with the advancement of quantum theory. In this work, we introduce several potential quantum-state texture measure schemes and check whether they satisfy the three fundamental conditions required for a valid quantum-state texture measure. Specifically, the measure induced by the $l_1$-norm serves as a vital tool for quantifying coherence, but we prove that it cannot be used to quantify quantum state texture. Furthermore, we show that while relative entropy and robustness meet three fundamental conditions, they are not optimal for quantifying quantum-state texture. Fortunately, we still find that there are several measures that can be used as the measure standard of quantum-state texture. Among them, the trace distance measure and the geometric measure are two good measurement schemes. In addition, the two measures based on Uhlmann's fidelity are experimentally friendly and can serve as an ideal definition of quantum-state texture measures in nonequilibrium situations. All these researches on quantum-state texture measure theory can enrich the resource theory framework of quantum-state texture.
	\end{abstract}
	
	 \maketitle
	
\section{I. Introduction}
The notion of quantum-state texture (QST) was first proposed by Parisio \cite{fp}, as a novel quantum resource intrinsically related to coherence. 
QST can be understood simply like this: Under the given computational basis $\{|i\rangle\}$, we can imagine the density matrix $\rho$ as a plot, and the matrix entry can be regarded as the altitude of this coordinate on the plot. We see the row and column indicators of the matrix as the first two dimensions, and the real part of each $\rho_{ij}$ and an analogous plot for the imaginary part, that is, the altitude, as the third dimension, so each quantum state corresponds to a three-dimensional plot. Under this perspective, the plot will generally show unevenness, or texture. Among them, the simplest possible plot is the one for which all altitudes are filled with the same number. The only existing quantum state which gives rise to such a plot is \begin{equation}\label{e1}
f_1=|f_1\rangle\langle f_1|,
\end{equation}
where $|f_1\rangle=\frac{1}{\sqrt{d}}\sum_{i=0}^{d-1}|i\rangle$ and $d$ is the dimension of Hilbert space $H$.  We refer to
this state as textureless and it alone corresponds to the
zero-resource set of the theory to be developed. Parisio \cite{fp} demonstrated that the circuit layer, which (contains at least one CNOT gate) can be fully characterized by randomized input states and the texture of the output qubits. Notably, this process does not require tomographic protocols or ancillary systems. Therefore, QST shares similarities with quantum coherence, which is a key component in emerging quantum technologies, such as quantum metrology \cite{vgsl,vgsll}, quantum computing \cite{mh}, nanoscale thermodynamics \cite{ggmp,mldj,vngg,gfjd}, and biological systems \cite{sl,sfhm,erra}. Therefore, we have sufficient reason to believe that QST, as the next quantum resource, will have a positive impact on quantum information tasks.
	
	An essential aspect of quantum resource theory is the quantification of quantum resources. As two fundamental quantum resources in quantum science and technology, entanglement and coherence have been the subject of numerous proposed measures \cite{rphk,wkw,vw,sk,fma,vvmb,tcwp,mcc,gy,jzyy,jfjq,gvrt,dsys,eb,srpr,ss,hhe,tbmc,asga,mlhx,kfbu,zxjs,fbhk,jxlh,llss,bhkd,cntr}. In \cite{rphk}, the authors summarize several methods for quantifying entanglement. A class of entanglement measures is based on the natural intuition that the closer a state is to the set of separable states, the less entangled it is. 
	\begin{equation}
	E(\rho)=\mathop{\inf}_{\sigma\in\mathcal{S}}D(\rho,\sigma),
	\end{equation}
	where the $\mathcal{S}$ is the set of separable states. Another way to define an entanglement measure is the convex roof method: we first propose an entanglement measure $E$ that holds for pure states and then extend it to general mixed states by convex roof:
	\begin{equation}
	E(\rho)=\mathop{\min}_{\{p_i,|\psi_i\rangle\}}\sum_{i}p_iE(|\psi_i\rangle),\sum_{i}p_i=1, p_i\geq0,
	\end{equation}
    where the minimum is taken over all possible pure state decomposition $\rho=\sum_{i}p_i|\psi_i\rangle\langle\psi_i|$. The minimum is reached for a particular ensemble called the optimal ensemble \cite{au1}. There are other types of entanglement measures, among which entanglement robustness \cite{gvrt} has been extensively studied and widely applied. Robustness can be interpreted as the ability to maintain certain properties even in the presence of noisy environments. From a chronological perspective, research on coherence measures lags behind entanglement, but the core concepts and construction methods for coherence measures closely resemble those for entanglement. Baumgratz et al. \cite{tbmc} established the conditions that coherence measures should satisfy and identified classes of functions that meet these conditions. Napoli et al. \cite{cntr} defined coherence robustness based on the concept of entanglement robustness, and although both definitions share the same foundational idea, they differ in certain aspects.

As a newly emerging quantum resource, it is natural to consider extending the definition methods used for entanglement and coherence measures to define a QST measure. In \cite{fp}, the author firstly introduced three basic conditions that a QST measure should satisfy. But similar to quantum entanglement and quantum coherence, it should have more quantization standards suitable for different geometric and physical meanings. Therefore, in this manuscript, we mainly study the quantification of QST. We define several measures based on the methods used for constructing entanglement and coherence measures and find that unlike in the cases of entanglement and coherence, $l_1$ measure, robustness and relative entropy are not suitable candidates for quantifying QST. 

The remainder of this paper is organized as follows. In Section 2, we introduce several QST measures and provide analytical lower bounds for the geometric measure of texture [Theorem 1-5]. Additionally, we explain why the three candidates of measure mentioned above are not appropriate QST measures in Appendix [Appendix A-C] and show that the trace distance measure we define offer advantages in distinguishing QST states when compared to existing QST measures for some cases. We summarize and discuss our conclusions in the final section.
     
\section{II. Quantum texture measures}

The basic requirements for a bona fide quantum state texture measure were introduced in \cite{fp}. Specifically, given an arbitrary QST measure, $\mathcal{T}$, it must satisfy the following conditions: 
(i) $\mathcal{T}(\rho)\geq0$, and $\mathcal{T}(f_1)=0$. (ii) $\mathcal{T}(\rho)\geq\mathcal{T}(\Lambda(\rho))$, where $\Lambda$ is completely positive and trace-preserving maps satisfied $\Lambda(f_1)=f_1$, are affected by Kraus operators, $\Lambda(\rho)=\sum_{n}K_n\rho K_n^\dagger$, for which $\sum_{n}K_n^\dagger K_n=\mathbbm{1}$. (iii) $\mathcal{T}$ is convex, i.e. $\mathcal{T}(\sum_{i}p_i\rho_i)\leq\sum_{i}p_i\mathcal{T}(\rho_i)$, $\sum_{i}p_i=1, \,p_i\geq0$. The author in \cite{fp} defined a state texture quantifier,
\begin{equation}\label{e2}
\mathfrak{R}(\rho)=-\ln\langle f_1|\rho|f_1\rangle,
\end{equation}
which they refer to as “state rugosity.” Quantum state texture, as a quantum resource, may also hold potential value due to its close connection with coherence. Thus, exploring possible candidates for QST measurement and analyzing them is an interesting and relevant problem.

A natural way to quantify QST of a state $\rho$ is through its distance to $f_1$ according to some distance measure $D$. Various QST measures correspond to different choices of $D$. A reasonable candidate for a QST measure is the trace distance, defined as:
\begin{equation}\label{e4}
\mathcal{D}(\rho,\sigma)=\frac{1}{2}\text{Tr}|\rho-\sigma|,
\end{equation}
where $\text{Tr}|A|=\text{Tr}(AA^\dagger)^\frac{1}{2}$ is the trace norm (or 1-norm) of $A$. We can then introduce a QST measure based on the trace distance, referred to as the trace distance measure of QST:
\begin{equation}\label{e5}
\mathcal{T}_\text{tr}(\rho)=\mathcal{D}(\rho,f_1).
\end{equation}

\begin{theorem}
The trace distance measure of QST $\mathcal{T}_\text{tr}$ is a well-defined texture measure.
\end{theorem}
\begin{proof}
For any quantum state $\rho$, it is evident that $\mathcal{T}_\text{tr}(\rho)\geq0$ and $\mathcal{T}_\text{tr}(f_1)=0$.

For completely positive and trace-preserving maps $\Lambda$, we have:
\begin{equation*}
\begin{split}
\mathcal{T}_\text{tr}(\Lambda(\rho))&=\mathcal{D}(\Lambda(\rho),f_1)\\
                               &=\mathcal{D}(\Lambda(\rho),\Lambda(f_1))\\
                               &\leq\mathcal{D}(\rho,f_1)\\
                               &=\mathcal{T}_\text{tr}(\rho),
\end{split}
\end{equation*}
where the inequality is according to Ref. \cite{mc}.

The authors in \cite{mc} have demonstrated the strong convexity of trace distance, which can be expressed as:
\begin{equation}\label{e6}
\mathcal{D}(\sum_{i}p_i\rho_i,\sum_{i}q_i\sigma_i)\leq D(p,q)+\sum_{i}p_i\mathcal{D}(\rho_i,\sigma_i),
\end{equation}
where $D$ represents the classical trace distance between probability distribution $\{p_i\}$ and $\{q_i\}$: $D(p,q)=\frac{1}{2}\sum_{i}|p_i-q_i|$.
Thus, we obtain:
\begin{equation*}
\begin{split}
\mathcal{T}_\text{tr}(\sum_{i}p_i\rho_i)&=\mathcal{D}(\sum_{i}p_i\rho_i,f_1)\\
                                   &=\mathcal{D}(\sum_{i}p_i\rho_i,\sum_{i}p_if_1)\\
                                   &\leq D(p,p)+\sum_{i}p_i\mathcal{D}(\rho_i,f_1)\\
                                   &=\sum_{i}p_i\mathcal{D}(\rho_i,f_1)\\
                                   &=\sum_{i}p_i\mathcal{T}_\text{tr}(\rho_i).
\end{split}
\end{equation*}
The inequality above is based on Eq. (\ref{e6}).
\end{proof}
The $l_1$ matrix norm serves as a common distance measure in quantum resource theory and plays a significant role in quantifying quantum resources such as coherence \cite{tbmc}. However, as we demonstrate in Appendix A, the measure induced by the $l_1$ norm is unsuitable for quantifying QST.

The degree of QST of a pure quantum state can be characterized by the distance or angle to the textureless state. Based on this insight, we propose the geometric measure of QST. For a pure state $|\psi\rangle$, the geometric measure of QST $\mathcal{T}_g(|\psi\rangle)$ is
\begin{equation}\label{e11}
\mathcal{T}_g(|\psi\rangle)=1-|\langle f_1|\psi\rangle|^2.
\end{equation}
The geometric measure for general mixed state $\rho$ is given by
convex roof extension,
\begin{equation}\label{e12}
	\mathcal{T}_g(\rho)=\mathop{\min}_{\{p_i,|\psi_i\rangle\}}
	\sum_{i}p_i\mathcal{T}_g(|\psi_i\rangle),
\end{equation}
where the minimization goes over all possible pure-state
decompositions of
$\rho=\sum_{i}p_i|\psi_i\rangle\langle\psi_i|$.
\begin{theorem}
The geometric measure is an eligible QST measure.
\end{theorem}
\begin{proof}
It is straightforward to verify that for any state $\rho$, $\mathcal{T}_g(\rho)\geq0$ and $\mathcal{T}_g(f_1)=0$.

Now, prove convexity. Consider $\rho=t\rho_1+(1-t)\rho_2$. Let $\rho_1=\sum_{i}p_i|\psi_i\rangle\langle\psi_i|$ and $\rho_2=\sum_{j}q_j|\phi_j\rangle\langle\phi_j|$ be the optimal pure-state decomposition of $\mathcal{T}_g(\rho_1)$ and $\mathcal{T}_g(\rho_2)$, respectively. Where $\sum_{i}p_i=\sum_{j}q_j=1$ and $p_i$, $q_j>0$. We then have
\begin{equation*}
\begin{split}
\mathcal{T}_g(\rho)&=\mathcal{T}_g(t\sum_{i}p_i|\psi_i\rangle\langle\psi_i|+(1-t)\sum_{j}q_j|\phi_j\rangle\langle\phi_j|)\\
                   &\leq \sum_{i}tp_i\mathcal{T}_g(|\psi_i\rangle\langle\psi_i|)+\sum_{j}(1-t)q_j\mathcal{T}_g(|\phi_j\rangle\langle\phi_j|)\\
                   &=t\mathcal{T}_g(\rho_1)+(1-t)\mathcal{T}_g(\rho_2),
\end{split}
\end{equation*}
where the inequality is due to that $\sum_{i}tp_i|\psi_i\rangle\langle\psi_i|+\sum_{j}(1-t)q_j|\phi_j\rangle\langle\phi_j|$ is also a pure state decomposition of $\rho$.

We demonstrate that the geometric measure of QST $\mathcal{T}_g$ cannot increase under completely positive and trace-preserving maps $\Lambda$ in the following. For any pure state $|\psi\rangle$, one have
\begin{equation*}
\begin{split}
\mathcal{T}_g(\Lambda(|\psi\rangle))&=\mathcal{T}_g(\sum_{n}K_n|\psi\rangle\langle\psi|K_n^\dagger)\\
                                    &\leq\sum_{n}p_n\mathcal{T}_g(K_n|\psi\rangle\langle\psi|K_n^\dagger/\text{Tr}(K_n|\psi\rangle\langle\psi|K_n^\dagger))\\
                                    &=\sum_{n}p_n\mathcal{T}_g(K_n|\psi\rangle/\sqrt{p_n})\\
                                    &=\sum_{n}p_n(1-\alpha_n^2|\langle f_1|\psi\rangle|^2/p_n)\\
                                    &=\mathcal{T}_g(|\psi\rangle),      
\end{split}
\end{equation*}
where $p_n=\text{Tr}(K_n|\psi\rangle\langle\psi|K_n^\dagger)$, the inequality follows from convexity, and the third equality arises because
\begin{equation*}
\begin{split}
\mathcal{T}_g(K_n|\psi\rangle/\sqrt{p_n})&=1-|\langle f_1|K_n|\psi\rangle|^2/p_n\\
                                         &=1-|\langle f_1|\alpha_n|\psi\rangle|^2/p_n\\
                                         &=1-\alpha_n^2|\langle f_1|\psi\rangle|^2/p_n,\\
\end{split}
\end{equation*}
where the second equality above follows from $f_1\propto K_n f_1$. Let $\rho=\sum_{i}p_i|\psi_i\rangle\langle\psi_i|$ be the optimal pure state decomposition of $\mathcal{T}_g(\rho)$, i.e. $\mathcal{T}_g(\rho)=\sum_{i}p_i\mathcal{T}_g(|\psi_i\rangle)$. Then, we have
\begin{equation*}
\begin{split}
\mathcal{T}_g(\Lambda(\rho))&=\mathcal{T}_g(\sum_{i}p_i\Lambda(|\psi_i\rangle))\\
                            &\leq\sum_{i}p_i\mathcal{T}_g(\Lambda(|\psi_i\rangle))\\
                            &\leq\sum_{i}p_i\mathcal{T}_g(|\psi_i\rangle)\\
                            &=\mathcal{T}_g(\rho).
\end{split}
\end{equation*}
\end{proof}

For general quantum states, the geometric measure of QST $\mathcal{T}_g$ is difficult to obtain directly, so it is useful to provide an analytical lower bound for $\mathcal{T}_g$. In fact, the connection between the geometric measure and the trace distance leads to a lower bound of $\mathcal{T}_g$.

From Eq. (\ref{e6}), we can easily conclude that $\mathcal{D}$ is jointly convex,
\begin{equation}\label{e13}
\mathcal{D}(\sum_{i}p_i\rho_i,\sum_{i}p_i\sigma_i)\leq\sum_{i}p_i\mathcal{D}(\rho_i,\sigma_i).
\end{equation}
Then, we have the following results.
\begin{theorem}
For any quantum state $\rho$, we have
\begin{equation}\label{e14}
\mathcal{T}_g(\rho)\geq[\mathcal{D}(\rho,f_1)]^2.
\end{equation}
\end{theorem}
\begin{proof}
Firstly, for any two pure states $|\psi\rangle$ and $|\phi\rangle$, one have \cite{pjcm}
\begin{equation}\label{e15}
\mathcal{D}(|\psi\rangle,|\phi\rangle)^2=\frac{1}{2}\||\psi\rangle\langle\psi|-|\phi\rangle\langle\phi|\|_F^2,
\end{equation}
where $\|\,.\,\|_F$ is the Frobenius norm. Note that for the pure state $f_1$,
\begin{equation*}
\begin{split}
\frac{1}{2}\||\psi\rangle\langle\psi|-|f_1\rangle\langle f_1|\|_F^2&=\text{Tr}[(|\psi\rangle\langle\psi|-|f_1\rangle\langle f_1|)^2]\\
          &=2-2\text{Tr}(|\psi\rangle\langle\psi|.|f_1\rangle\langle f_1|)\\
          &=2-2|\langle\psi|f_1\rangle|^2.
\end{split}
\end{equation*}
Substituting the above equation into Eq. (\ref{e15}) yields
\begin{equation}\label{e16}
\mathcal{D}(|\psi\rangle,f_1)^2=\mathcal{T}_g(|\psi\rangle).
\end{equation}
Thus for general state $\rho$, we have
\begin{equation*}
\begin{split}
\mathcal{T}_g(\rho)&=\sum_{i}p_i\mathcal{T}_g(|\psi_i\rangle)\\
                   &=\sum_{i}p_i\mathcal{D}(|\psi_i\rangle,f_1)^2\\
                   &\geq[\sum_{i}p_i\mathcal{D}(|\psi_i\rangle,f_1)]^2\\
                   &\geq\mathcal{D}(\rho,f_1)^2,
\end{split}
\end{equation*}
where $\rho=\sum_{i}p_i|\psi_i\rangle\langle\psi_i|$ is the optimal pure state decomposition corresponding to $\mathcal{T}_g$, the second equation is due to (\ref{e16}), and the last inequality is according to the convexity of $\mathcal{D}$.
\end{proof}

Uhlmann's Fidelity between two quantum states is defined as \cite{au}
\begin{equation}\label{e17}
F(\rho,\sigma)=[\text{Tr}(\sqrt{\rho^{1/2}\sigma\rho^{1/2}})]^2.
\end{equation}
In particular, if one state is pure, $\sigma=|\psi\rangle\langle\psi|$, then $F(\rho,\sigma)=\langle\psi|\rho|\psi\rangle$. Hence, the QST measure induced by fidelity is defined as
\begin{equation}\label{e18}
\mathcal{T}_F(\rho)=1-F(\rho,f_1).
\end{equation}
\begin{theorem}
$\mathcal{T}_F$ is a faithful quantum-state texture measure.
\end{theorem}
\begin{proof}
Note that $f_1=|f_1\rangle\langle f_1|$ is pure, so we have
\begin{equation*}
\begin{split}
\mathcal{T}_F(\rho)&=1-F(\rho,f_1)\\
                   &=1-\langle f_1|\rho|f_1\rangle\\
                   &=1-\sum_{i,j}\rho_{ij}/d,
\end{split}
\end{equation*}
where $d$ is the dimensional of a quantum system. For condition (i), we have $\mathcal{T}_F(\rho)\geq0$ for any $\rho$ and $\mathcal{T}_F(f_1)=0$.

For condition (ii), the author proved in Ref. \cite{fp} that
\begin{equation}\label{e19}
\sum_{i,j}\rho_{ij}\leq\sum_{i,j}\Lambda(\rho)_{ij}.
\end{equation}
According to Eq. (\ref{e19}), we can obtain
\begin{equation*}
\mathcal{T}_F(\Lambda(\rho))=1-\sum_{i,j}\Lambda(\rho)_{ij}/d\leq1-\sum_{i,j}\rho_{ij}/d=\mathcal{T}_F(\rho).
\end{equation*}

For convexity, let $\rho=t\sigma+(1-t)\tau$, one have
\begin{equation*}
\begin{split}
\mathcal{T}_F(\rho)&=\mathcal{T}_F(t\sigma+(1-t)\tau)\\
                   &=1-\sum_{i,j}(t\sigma_{ij}+(1-t)\tau_{ij})/d\\
                   &=t-t\sum_{i,j}\sigma_{ij}/d+(1-t)-(1-t)\sum_{i,j}\tau_{ij}/d\\
                   &\leq t\mathcal{T}_F(\sigma)+(1-t)\mathcal{T}_F(\tau).
\end{split}
\end{equation*}
\end{proof}

From the proof in Theorem 4, we observe that $F(\rho,f_1)=\text{Tr}(\rho f_1)$, which indicates that $\mathcal{T}_F$ is an experimentally measurable QST measure. In addition, we can also define the Bures measure \cite{vvmb} of QST based on fidelity
\begin{equation*}
\mathcal{T}_B(\rho)=2(1-\sqrt{F(\rho,f_1)}).
\end{equation*}
Like $\mathcal{T}_F$, it is also an experimentally friendly QST measure.

\begin{theorem}
The Bures measure $\mathcal{T}_B$ is a well-defined QST measure.
\end{theorem}
\begin{proof}
First of all, since $F(\rho,f_1)=1$ if and only if $\rho=f_1$ and $0\leq F(\rho_1,\rho_2)\leq1$ for any states $\rho_1$ and $\rho_2$ \cite{jamz}, we have $\mathcal{T}_B(f_1)=0$ and $\mathcal{T}_B(\rho)\geq0$.

Secondly, based on the proof in Theorem 5, we obtain that $F(\Lambda(\rho),f_1)\geq F(\rho,f_1)$. Using monotonicity, it is easy to conclude that
\begin{small}
\begin{equation*}
\mathcal{T}_{B}(\Lambda(\rho))=2(1-\sqrt{F(\Lambda(\rho),f_1)})\leq2(1-\sqrt{F(\rho,f_1)})=\mathcal{T}_{B}(\rho).
\end{equation*}
\end{small}

Finally, we prove convexity. It is known that $\sqrt{F}$ is jointly concave \cite{jamz}, i.e.
\begin{equation}\label{e20}
\begin{split}
&\sqrt{F(t\rho_1+(1-t)\rho_2,t\sigma_1+(1-t)\sigma_2)}\geq t\sqrt{F(\rho_1,\sigma_1)}\\
&+(1-t)\sqrt{F(\rho_2,\sigma_2)}.
\end{split}
\end{equation}
Then we have
\begin{footnotesize}
\begin{equation*}
	\begin{split}
		\mathcal{T}_B(t\rho_1+(1-t)\rho_2)&=2(1-\sqrt{F(t\rho_1+(1-t)\rho_2,f_1)})\\
		&=2(1-\sqrt{F(t\rho_1+(1-t)\rho_2,tf_1+(1-t)f_1)})\\
		&\leq2(1-t\sqrt{F(\rho_1,f_1)}-(1-t)\sqrt{F(\rho_2,f_1)})\\
		&=t.2(1-\sqrt{F(\rho_1,f_1)})+(1-t).2(1-\sqrt{F(\rho_2,f_1)})\\
		&=t\mathcal{T}_B(\rho_1)+(1-t)\mathcal{T}_B(\rho_2).
	\end{split}
\end{equation*}
\end{footnotesize}
\end{proof}

In fact, like state rugosity in \cite{fp}, the QST measures $\mathcal{T}_F$ and $\mathcal{T}_B$ above can also serve as indicators for a nonequilibrium situation. A general system in thermal equilibrium with a reservoir at absolute temperature $T$ is being considered by us, which can be described by the canonical Gibbs state,
\begin{equation}\label{e21}
\varrho(T)=\frac{1}{\mathcal{Z}}\sum_{i}e^{-E_i/k_BT}|i\rangle\langle i|,
\end{equation}
where $\mathcal{Z}$ is the canonical partition functional and $k_B$ is the Boltzmann constant. The calculation shows that
\begin{equation*}
\begin{split}
&\mathcal{T}_F(\varrho(T))=\frac{d-1}{d},\\
&\mathcal{T}_B(\varrho(T))=\frac{2(d-\sqrt{d})}{d},
\end{split}
\end{equation*}
where we used the fact that $\sum_{i}\frac{e^{-E_i/k_BT}}{\mathcal{Z}}=1$ because $\varrho(T)$ is a diagonal state. This means that $\mathcal{T}_F$ and $\mathcal{T}_B$ only depend on the dimensionality $d$ here, and are independent of temperature. Therefore, in this case, the variation of QST with temperature $T$ can serve as a witness of an out-of-equilibrium situation.

Consider the coherent Gibbs kets:
\begin{equation}
|\psi\rangle_T=\frac{1}{\sqrt{\mathcal{Z}}}\sum_{i}e^{-E_i/2k_BT}|i\rangle,
\end{equation}
which is a class of nonequilibrium states. Now use our measures to calculate the QST of $|\psi\rangle_T$, it can be obtained that
\begin{equation*}
\begin{split}
&\mathcal{T}_F(|\psi\rangle_T)=\frac{d-1}{d}-\frac{1}{d\mathcal{Z}}\sum_{i\neq j}e^{E_j-E_i/2k_BT},\\
&\mathcal{T}_B(|\psi\rangle_T)=2-2\sqrt{\frac{1}{d\mathcal{Z}}\sum_{i,j}e^{E_j-E_i/2k_BT}}.
\end{split}
\end{equation*}
So if $T_1\neq T_2$, then $\mathcal{T}_F(|\psi\rangle_{T_1})\neq\mathcal{T}_F(|\psi\rangle_{T_2})$,  and $\mathcal{T}_B(|\psi\rangle_{T_1})\neq\mathcal{T}_B(|\psi\rangle_{T_2})$.

In addition to the QST measures given above, we can also consider the relative entropy and robustness. In Appendix B, we show that although the relative entropy measure $\mathcal{T}_r(\rho)$ satisfies the above three conditions, it is not a suitable candidate for quantifying QST. Moreover, we have proved in Appendix C that the robustness $\mathcal{T}_R(\rho)$ also meets these conditions, but it is also not a suitable candidate for quantifying QST. Both schemes have a big disadvantage, that is, the degree of distinction between different quantum states is too low. Look at the following concrete example.

\mathbi{Example 1.} Let us consider the two-qubit Bell states $|\psi^+\rangle=(|00\rangle+|11\rangle)/\sqrt{2}$ and  $|\psi^-\rangle=(|00\rangle-|11\rangle)/\sqrt{2}$. We now use our defined QST measure to calculate their quantum-state texture.

By simple calculation, one has:
\begin{equation*}
\begin{split}
&\mathcal{T}_{tr}(|\psi^+\rangle)=\frac{\sqrt{2}}{2},\,\mathcal{T}_{tr}(|\psi^-\rangle)=1;\\
&\mathcal{T}_g(|\psi^+\rangle)=\frac{1}{2},\,\mathcal{T}_g(|\psi^-\rangle)=1;\\
&\mathcal{T}_F(|\psi^+\rangle)=\frac{1}{2},\,\mathcal{T}_F(|\psi^-\rangle)=1;\\
&\mathcal{T}_B(|\psi^+\rangle)=2-\sqrt{2},\,\mathcal{T}_B(|\psi^-\rangle)=2;\\
&\mathcal{T}_r(|\psi^+\rangle)=+\infty,\,\mathcal{T}_r(|\psi^-\rangle)=+\infty;\\
&\mathcal{T}_R(|\psi^+\rangle)=+\infty,\,\mathcal{T}_R(|\psi^-\rangle)=+\infty;\\
&\mathfrak{R}(|\psi^+\rangle)=\ln 2,\,\mathfrak{R}(|\psi^-\rangle)=+\infty.
\end{split}
\end{equation*}
Through this concrete example, we can also clearly see that neither $\mathcal{T}_r(\rho)$ nor $\mathcal{T}_R(\rho)$ is suitable as a measure of the QST. In fact, please refer to our theoretical analysis in Appendix B and Appendix C for their unsuitability as quantitative standards for QST.

Additionally, we need to pay attention to a fact: Fourier states are proved to be maximal in a formal quantum state texture resource theoretical sense \cite{fp},
\begin{equation}
|f_k\rangle=\frac{1}{\sqrt{d}}\sum_{j=1}^{d}\omega_d^{(k-1)(j-1)}|j\rangle,
\end{equation}
where $\omega_d=e^{2\pi i/d}$ and $k=2,\dots,d$. Therefore, all proper measures should attain their maximum values for the Fourier states $|f_k\rangle$ with $k>1$. This is indeed consistent with condition (ii), since any other quantum state can be obtained from these Fourier states through some free operations \cite{fp}. Furthermore, state rugosity $\mathfrak{R}$ attains its maximum value not only for the Fourier states $|f_k\rangle\,(k>1)$, but also for their linear combinations. This raises an interesting question: whether the proposed measures yield different results for these linear combinations compared to the Fourier states where they achieve maximality. In Appendix D, we show that the proposed measures attain their maximum values for both Fourier states $|f_k\rangle\,(k>1)$ and their linear combinations.

We demonstrate the important role of different QST measures with the following example. Let us consider two single-parameter quantum state families.

\mathbi{Example 2.} Given two state families
$$\sigma_\alpha=\frac{1}{4}\left(
\begin{array}{cccc}
	1 & 0 & 0 & \alpha\\
	0 & 1 & \alpha & 0\\
	0 & \alpha & 1 & 0\\
	\alpha & 0 & 0 & 1\\
\end{array} 
\right ),$$
and
$$\tau_\alpha=\frac{1}{2}\left(
\begin{array}{cccc}
	1 & 0 & 0 & \alpha\\
	0 & 0 & 0 & 0\\
	0 & 0 & 0 & 0\\
	\alpha & 0 & 0 & 1\\
\end{array} 
\right ),$$
with $0\leq\alpha\leq1$.

The calculation shows that $\mathfrak{R}(\sigma_\alpha)=\mathfrak{R}(\tau_\alpha)=-\ln(\frac{1+\alpha}{4})$. In other words, the state rugosity given by \cite{fp} cannot distinguish between these two states. However, for trace distance, we have
\begin{equation*}
	\begin{split}
	&\mathcal{T}_\text{tr}(\sigma_\alpha)=\frac{3-\alpha}{4},\\
	&\mathcal{T}_\text{tr}(\tau_\alpha)=\frac{1}{4}(1-\alpha+\sqrt{\alpha^2+2\alpha+5}).
	\end{split}
\end{equation*}
Clearly, $\mathcal{T}_\text{tr}(\tau_\alpha)>\mathcal{T}_\text{tr}(\sigma_\alpha)$, and the corresponding numerical results are shown in Figure 1, which shows that our trace distance measurement can distinguish all $\sigma_\alpha$ and $\tau_\alpha$ for parameters $\alpha \in[0,1]$. However, we need to point out that even if two measures $M_1$ and $M_2$ satisfy all the necessary conditions, there could also be two quantum states $\sigma$ and $\tau$ such that $M_1(\tau)\geq M_1(\sigma)$ and $M_2(\sigma)\geq M_2(\tau)$. That is to say, comparing two measures is scenario-dependent.
\begin{figure}[htbp]
	\centering
	\includegraphics[width=0.5\textwidth]{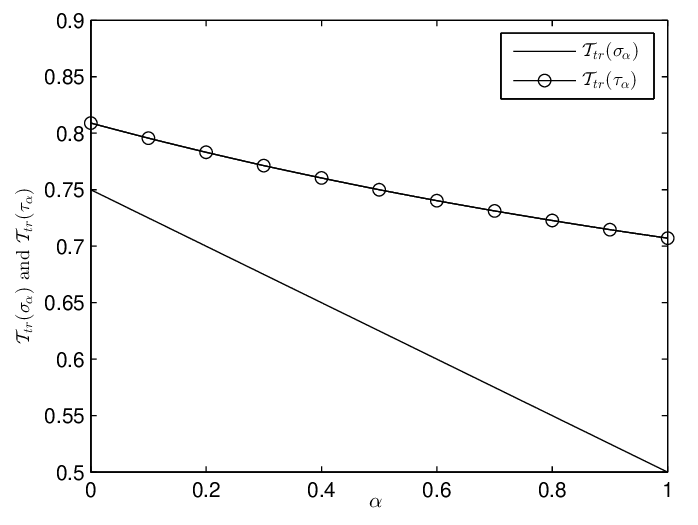}
	\vspace{-1em} \caption{For different values of $\alpha$, the values of  $\mathcal{T}_\text{tr}(\sigma_\alpha)$ and $\mathcal{T}_\text{tr}(\tau_\alpha)$ as shown in figure.} \label{Fig.1}
\end{figure}

\section{III. Conclusions and discussions}
We proposed several definitions of QST measures and theoretically demonstrated that they satisfy the three fundamental conditions that a quantum-state texture measure should fulfil, as outlined in \cite{fp}. Among the measures we define, the trace distance measure $\mathcal{T}_\text{tr}$ offers significant advantages in distinguishing texture states compared to the existing QST measure for some states. In addition, the geometric measure of QST $\mathcal{T}_g$, constructed using the convex roof method, faces the challenge of being difficult to compute for general mixed states. To address this, we provide an analytical lower bound for $\mathcal{T}_g$. Furthermore, the two measures related to fidelity $\mathcal{T}_F$ and $\mathcal{T}_B$ are experimentally friendly and can serve as reliable indicators of a nonequilibrium situation. 

Unlike quantum entanglement and quantum coherence, we demonstrated that the relative entropy $\mathcal{T}_r$ and robustness $\mathcal{T}_R$ are not suitable for quantifying quantum-state texture, even though they satisfy the three basic conditions of QST measures. Additionally, we have demonstrated that while the measure induced by the $l_1$-norm serves as a vital tool for quantifying coherence, it cannot be used to quantify quantum-state texture. We firmly believe that the work presented in this manuscript contributes to the enrichment and advancement of QST measurement and provides a solid theoretical foundation for the further development of QST theory.

\bigskip
{\bf Acknowledgments:} ~This work is supported by the National Natural Science Foundation of China (NSFC) under Grant No. 12171044 and the specific research fund of the Innovation Platform for Academicians of Hainan Province.

\appendix

\section{Appendix A: the $l_1$ measure of QST}
\setcounter{equation}{0}
\renewcommand{\theequation}{A\arabic{equation}}
Analogous to the case of quantum coherence, the $l_1$-norm-induced measure for quantum state texture can be expressed as:
\begin{equation}
	\mathcal{T}_1(\rho)=\|\rho-f_1\|_{l_1},
\end{equation}
where $\|\,.\,\|_{l_1}$ is the $l_1$ matrix norm such that $\|\rho\|_{l_1}=\sum_{i,j}|\rho_{ij}|$. 

We can directly verify condition (i): $\mathcal{T}_1(f_1)=\|f_1-f_1\|_{l_1}=0$ and $\mathcal{T}_1(\rho)\geq0$ for any state $\rho$ based on the definition of $l_1$ norm. For convexity, given $\rho=t\sigma+(1-t)\tau$, we have:
\begin{equation*}
	\begin{split}
		\mathcal{T}_1(t\sigma+(1-t)\tau)&=\|t\sigma+(1-t)\tau-f_1\|_{l_1}\\
		&=\sum_{i,j}|t\sigma_{ij}+(1-t)\tau_{ij}-\frac{1}{d}|\\
		&\leq\sum_{i,j}t|\sigma_{ij}-\frac{1}{d}|+(1-t)|\tau_{ij}-\frac{1}{d}|\\
		&=t\mathcal{T}_1(\sigma)+(1-t)\mathcal{T}_1(\tau),
	\end{split}
\end{equation*}
where $d$ is the dimensional of the Hilbert space. Thus far, we have shown that $\mathcal{T}_1$ satisfies conditions (i) and (iii) for a QST measure. While the satisfaction of condition (ii) appeared certain, our counterexample unfortunately demonstrates that $\mathcal{T}_1$ fails to meet condition (ii).

Consider the case of dimension $d=2$, where there is only one maximum QST state, i.e
\begin{equation}
	|f_2\rangle=\frac{1}{\sqrt{2}}|0\rangle-\frac{1}{\sqrt{2}}|1\rangle.
\end{equation}
The simple calculation shows that $\mathcal{T}_1(|f_2\rangle)=2$. But if one consider the state
\begin{equation}
	|f^*\rangle=\frac{\sqrt{3}}{2}|0\rangle-\frac{1}{2}|1\rangle,
\end{equation}
we have $\mathcal{T}_1(|f^*\rangle)=\frac{3+\sqrt{3}}{2}>2=\mathcal{T}_1(|f_2\rangle)$. The author demonstrated in \cite{fp} that any other quantum state can be reached from the states with maximal QST resources by applying some free operation. In other words, there must exist a free map $\Lambda$ satisfying $\Lambda(|f_2\rangle)=|f^*\rangle$. However, $\mathcal{T}_1(|f^*\rangle)>\mathcal{T}_1(|f_2\rangle)$ implies that the $l_1$ measure can increase under free operation, which violates condition (ii). Consequently, the measure induced by the $l_1$ matrix norm cannot be considered a proper QST measure.
\\

\section{Appendix B: the relative entropy of QST}
\setcounter{equation}{0}
\renewcommand{\theequation}{B\arabic{equation}}
In this section, we explain why relative entropy is not a suitable QST measure.
Similar to entanglement and coherence, we can propose the relative entropy measure of QST:
\begin{equation}\label{e7}
	\mathcal{T}_r(\rho)=S(\rho\|f_1),
\end{equation}
where $S$ is the quantum relative entropy, $S(\rho\|\sigma)=Tr(\rho\log(\rho)-\rho\log(\sigma)).$

It is clear that $\mathcal{T}_r(f_1)=S(f_1\|f_1)=0$ and $\mathcal{T}_r(\rho)=S(\rho\|f_1)\geq0$ for any state $\rho$.

Then we prove convexity. In \cite{mbr}, the authors prove that the relative entropy is jointly convex, that is, if $\rho=\sum_{i}p_i\rho_i$ and $\sigma=\sum_{i}p_i\sigma_i,\,\sum_{i}p_i=1, p_i\geq0$, then 
\begin{equation}\label{e8}
	S(\rho\|\sigma)\leq\sum_{i}p_i S(\rho_i\|\sigma_i).
\end{equation}
Thus, we have
\begin{equation*}
	\begin{split}
		\mathcal{T}_r(\sum_{i}p_i\rho_i)&=S(\sum_{i}p_i\rho_i\|f_1)\\
		&=S(\sum_{i}p_i\rho_i\|\sum_{i}p_i f_1)\\
		&\leq\sum_{i}p_i S(\rho_i\|f_1)\\
		&=\sum_{i}p_i\mathcal{T}_r(\rho_i),
	\end{split}
\end{equation*}
where the inequality is due to Eq. (\ref{e8}).

Before proving the monotonicity under completely positive and trace-preserving maps, let us first prove that 
\begin{equation}\label{e9}
	\mathcal{T}_r(\rho)\geq\sum_{n}q_n\mathcal{T}_r(\rho_n),
\end{equation}
where $\rho_n=K_n\rho K_n^\dagger/q_n, q_n=Tr(K_n\rho K_n^\dagger)$, $\{K_n\}$ are Kraus operators with $\sum_{n}K_n K_n^\dagger=\mathbbm{1}$. We have
\begin{equation*}
	\begin{split}
		\mathcal{T}_r(\rho)&=S(\rho\|f_1)\\
		&\geq\sum_{n}q_n S(\rho_n\|K_n f_1 K_n^\dagger/tr(K_n f_1 K_n^\dagger))\\
		&=\sum_{n}q_n S(\rho_n\|f_1)\\
		&=\sum_{n}q_n \mathcal{T}_r(\rho_n),
	\end{split}
\end{equation*}
where the inequality above is according to the \cite{tbmc} and the second equation is due to $K_n|f_1\rangle\propto|f_1\rangle$ for all $n$ \cite{fp}. Specifically, we can set that $K_n|f_1\rangle=\alpha|f_1\rangle\,(\alpha>0)$, then $K_n f_1 K_n^\dagger=\alpha^2 f_1$. Thus one have $K_n f_1 K_n^\dagger/tr(K_n f_1 K_n^\dagger)=f_1$. Finally, according to Eq. (\ref{e9}) and convexity, we can obtain that
\begin{equation}\label{e10}
	\mathcal{T}_r(\rho)\geq\sum_{n}q_n\mathcal{T}_r(\rho_n)\geq\mathcal{T}_r(\Lambda(\rho)).
\end{equation}

Although the relative entropy defined above meets the three requirements of QST measure, we believe that relative entropy is not a suitable QST measure. In fact, if there is a non-zero state vector that belongs to the intersection of $\text{supp}(\rho)$ and $\text{kernel}(\sigma)$, then infinity will be generated when calculating $\text{Tr}(\rho\log\sigma)$ \cite{mc}. The support set is the row space of the density matrix, and kernel space is the orthogonal complement of the row space. Note that the dimension of row space is equal to the rank of density matrix, so we can obtain a sufficient condition that the relative entropy is infinity: $S(\rho\|\sigma)=\infty$, if $\text{rank}(\rho)>\text{rank}(\sigma)$. $\text{rank}(f_1)=1$, so at least for any quantum state $\rho$ with rank greater than $1$, we have $S(\rho\|f_1)=\infty$. In other words, the relative entropy is not a proper measure to quantify QST in a quantum system.
\\
\section{Appendix C: quantum state texture robustness}
\setcounter{equation}{0}
\renewcommand{\theequation}{C\arabic{equation}}
Following the idea of entanglement robustness \cite{gvrt} and coherence robustness \cite{cntr}, it is natural to define QST robustness as:
\begin{equation}\label{e22}
	\mathcal{T}_R(\rho)=\mathop{\min}_{\sigma}\{s\geq0|\frac{\rho+s\sigma}{1+s}=f_1\},
\end{equation}
where the minimum is taken over all possible quantum states $\sigma$ such that it is convex combination with $\rho$ results in a texture-free state. We can prove that the robustness defined in this way satisfies the three conditions of a QST measure.

The conditions $\mathcal{T}_R(\rho)\geq0$ and $\mathcal{T}_R(f_1)=0$ are obvious.

We take the optimal pseudomixture for $\rho$ to be:
\begin{equation}\label{e23}
	\rho=(1+s)f_1-s\sigma_{opt}.
\end{equation}
For completely positive and trace-preserving maps $\Lambda$ where these free maps must not create QST, that is to say, $\Lambda(f_1)=f_1$. Apply the maps $\Lambda$ to both sides in (\ref{e23}), one have
\begin{equation*}
	\begin{split}
		\Lambda(\rho)&=(1+s)\Lambda(f_1)-s\Lambda(\sigma_{opt})\\
		&=(1+s)f_1-s\Lambda(\sigma_{opt}).
	\end{split}
\end{equation*}
The last equation above means that $\frac{\Lambda(\rho)+s\Lambda(\sigma_{opt})}{1+s}$ is a textureless state. According to the definition of robustness, we have
\begin{equation}
	\mathcal{T}_R(\Lambda(\rho))\leq s=\mathcal{T}_R(\rho).
\end{equation}

Now, we prove convexity. Let $\rho_1$ and $\rho_2$ are two states, and write for each the optimal pseudomixture $\rho_k=(1+s_k)f_1-s_k\sigma_k\,(k=1,2)$. Taking the convex combination $\rho=t\rho_1+(1-t)\rho_2$ with $0\leq t\leq1$, we have:
\begin{footnotesize}
	\begin{equation*}
		\begin{split}
			\rho&=t\rho_1+(1-t)\rho_2\\
			&=t(1+s_1)f_1-ts_1\sigma_1+(1-t)(1+s_2)f_1-(1-t)s_2\sigma_2\\
			&=[1+ts_1+(1-t)s_2]f_1-[ts_1\sigma_1+(1-t)s_2\sigma_2]\\
			&=[1+ts_1+(1-t)s_2]f_1-(ts_1+(1-t)s_2)\frac{[ts_1\sigma_1+(1-t)s_2\sigma_2]}{ts_1+(1-t)s_2}.
		\end{split}
	\end{equation*}
\end{footnotesize}
By definition,
$$\mathcal{T}_R(\rho)=s\leq ts_1+(1-t)s_2=t\mathcal{T}_R(\rho_1)+(1-t)\mathcal{T}_R(\rho_2),$$
which proves convexity. The QST robustness we define satisfies the three conditions. However, when $\rho$ belongs to the orthogonal support of $f_1$, then $\mathcal{T}_R(\rho)$ returns $\infty$. Because there is no finite value of $s$ and quantum state that make the equation $\frac{\rho+s\sigma}{1+s}=f_1$ hold.
\\
\section{Appendix D: proposed measures and Fourier States}
\setcounter{equation}{0}
\renewcommand{\theequation}{D\arabic{equation}}
\subsection{trace distance measure $\mathcal{T}_\text{tr}$}

For the general dimension $d$, $\sum_{j=2}^{n}p_jf_j\,(n\leq d)$ is a linear combination of Fourier states, the matrix of $f_1-\sum_{j=2}^{n}p_jf_j$ is expressed as
\begin{widetext}
$$
\frac{1}{d}\left(\begin{array}{ccccc}
	0 & 1-\sum_{j=2}^{n}p_j\overline{\omega}_d^{j-1} & 1-\sum_{j=2}^{n}p_j\overline{\omega}_d^{2(j-1)} & \cdots & 1-\sum_{j=2}^{n}p_j\overline{\omega}_d^{(d-1)(j-1)} \\
	1-\sum_{j=2}^{n}p_j\omega_d^{j-1} & 0 & 1-\sum_{j=2}^{n}p_j\overline{\omega}_d^{j-1} & \cdots & 1-\sum_{j=2}^{n}p_j\overline{\omega}_d^{(d-2)(j-1)} \\
	1-\sum_{j=2}^{n}p_j\omega_d^{2(j-1)} & 1-\sum_{j=2}^{n}p_j\omega_d^{j-1} & 0 & \cdots & 1-\sum_{j=2}^{n}p_j\overline{\omega}_d^{(d-3)(j-1)} \\
	\vdots & \vdots & \vdots & \ddots & \vdots \\
	1-\sum_{j=2}^{n}p_j\omega_d^{(d-1)(j-1)} & 1-\sum_{j=2}^{n}p_j\omega_d^{(d-2)(j-1)} & 1-\sum_{j=2}^{n}p_j\omega_d^{(d-3)(j-1)} & \cdots & 0
\end{array}\right),
$$
\end{widetext}
where $\omega_d=e^{2\pi i/d}$ and $\overline{\omega}_d$ is the complex conjugate of $\omega_d$. It can be verified that the eigenvalues of $f_1-\sum_{j=2}^{n}p_jf_j$ are $1, -p_2, -p_3,\dots,-p_n, 0,\dots,0$, and the number of zero depends on the value of $n$. Note that $f_1-\sum_{j=2}^{n}p_jf_j$ is Hermitian, so we have
\begin{equation}
\begin{split}
\text{Tr}|f_1-\sum_{j=2}^{n}p_jf_j|&=1+|-p_2|+\dots+|-p_n|=1+\sum_{j=2}^{n}p_j\\
&=2,
\end{split}
\end{equation}
which means that $\mathcal{T}_\text{tr}(\sum_{j=2}^{n}p_jf_j)=1$. $\mathcal{T}_\text{tr}$ reaches the maximum for both the Fourier states and their linear combinations when we realized a fact that for any two states $\rho_1$ and $\rho_2$, $\text{Tr}|\rho_1-\rho_2|\leq2$ holds.

\subsection{geometric measure $\mathcal{T}_g$}
Since the Fourier states  $f_k\,(k>2)$ belong to the orthogonal support of $f_1$, it is easy to obtain that
$$
\mathcal{T}_g(f_k)=1-|\langle f_1|f_k\rangle|^2=1-0=1,\,k>2.
$$
On the one hand, for any linear combination of Fourier states $\sum_{j=2}^{n}p_jf_j\,(n\leq d)$ such that $\sum_{j=2}^{n}p_j=1$, we have
\begin{equation}
\mathcal{T}_g(\sum_{j=2}^{n}p_jf_j)\leq\sum_{j=2}^{n}p_j\mathcal{T}_g(f_j)=\sum_{j=2}^{n}p_j=1,
\end{equation}
where inequality is due to the convexity of $\mathcal{T}_g$. On the other hand, combining Theorem 3 and Eq.\,(D1), one can obtain
\begin{equation}
\mathcal{T}_g(\sum_{j=2}^{n}p_jf_j)\geq\mathcal{D}(\sum_{j=2}^{n}p_jf_j,f_1)^2=1.
\end{equation}
Eq.\,(D2) and (D3) yield $\mathcal{T}_g(\sum_{j=2}^{n}p_jf_j)=1$.
\\
\subsection{fidelity measure $\mathcal{T}_F$ and Bures measure $\mathcal{T}_B$}
In this section, we only need to verify that for any linear combination of Fourier states $\sum_{j=2}^{n}p_jf_j$, the fidelity satisfies the following equation,
\begin{equation}
F(\sum_{j=2}^{n}p_jf_j,f_1)=0.
\end{equation}
$f_1$ is pure, so one have
\begin{equation*}
\begin{split}
F(\sum_{j=2}^{n}p_jf_j,f_1)&=\langle f_1|\sum_{j=2}^{n}p_jf_j|f_1\rangle\\
&=\sum_{j=2}^{n}p_j\langle f_1|f_j\rangle\langle f_j|f_1\rangle\\
&=0,
\end{split}
\end{equation*}
which completes the proof.

\end{document}